\documentclass[submission,copyright,creativecommons]{eptcs}
\providecommand{\event}{QPL 2019} 
\usepackage{breakurl}             

\usepackage{amssymb}
\usepackage{amsthm}

\usepackage{mathtools}
\usepackage{ stmaryrd }
\usepackage{cite}

\newtheorem{thm}{Theorem}[section]

\newtheorem{lem}[thm]{Lemma}

\newtheorem{remark}[thm]{Remark}
\usepackage{hyperref} 
\usepackage[utf8]{inputenc}
\usepackage{enumerate}
\usepackage{amsmath}
\usepackage{amssymb}
\usepackage[T1]{fontenc} 
\usepackage[utf8]{inputenc} 

\usepackage{color,soul}
\definecolor{ItalianApricot}{rgb}{1,0.7,0.5}
\sethlcolor{ItalianApricot}

\begin{document}

\newtheorem{defn}[thm]{Definition}
\newtheorem{example}[thm]{Example}
\newtheorem{query}[thm]{Question}

\numberwithin{equation}{section}

\numberwithin{equation}{section}

%
%

\renewcommand{\epsilon}{\varepsilon}
\renewcommand{\phi}{\varphi}
\renewcommand{\setminus}{\smallsetminus}

\newcommand\+[1]{\mathcal{#1}}
\renewcommand\*[1]{\mathbf{#1}}

\newcommand{\seq}[1]{{\left\langle{#1}\right\rangle}}

\newcommand{\N}{\mathbb{N}}
\newcommand{\Q}{\mathbb{Q}}
\newcommand{\R}{\mathbb{R}}
\newcommand{\Z}{\mathbb{Z}}

%
%

%

\title{Generating Randomness from a Computable, Non-random Sequence of Qubits}
\author{Tejas Bhojraj
\institute{}
\institute{Department of Mathematics\\
University of Wisconsin-Madison
WI,USA}
\email{bhojraj@wisc.edu}
 \and
\quad\qquad 
\institute{}
\email{\quad  \quad\qquad }
}
\def\titlerunning {Generating Randomness from a Computable, Non-random Sequence of Qubits}
\def\authorrunning{Tejas Bhojraj}

\maketitle

\begin{abstract}
Nies and Scholz \cite{unpublished} introduced the notion of a \emph{state} to describe an infinite sequence of qubits and defined quantum-Martin-L{\"o}f randomness for states, analogously to the well known concept of Martin-L{\"o}f randomness for elements of Cantor space (the space of infinite sequences of bits). We formalize how `measurement' of a state in a basis induces a probability measure on Cantor space. A state is `measurement random' (mR) if the measure induced by it, under any computable basis, assigns probability one to the set of Martin-L{\"o}f randoms. Equivalently, a state is mR if and only if measuring it in any computable basis yields a Martin-L{\"o}f random with probability one. While quantum-Martin-L{\"o}f random states are mR, the converse fails: there is a mR state, $\rho$ which is not quantum-Martin-L{\"o}f random. In fact, something stronger is true. While $\rho$ is computable and can be easily constructed, measuring it in any computable basis yields an arithmetically random sequence with probability one. I.e., classical arithmetic randomness can be generated from a computable, non-quantum random sequence of qubits. 
\end{abstract}

\section{Introduction}
Martin-L{\"o}f randomness for infinite sequences of bits, a central concept in algorithmic randomness\cite{misc}, has recently been generalized to infinite sequences of qubits \cite{unpublished}. As is standard in the literature in algorithmic randomness and mathematical logic, let $2^{\omega}$ denote the collection of infinite sequences of bits, let $2^n$ denote the set of bit strings of length $n$, $2^{<\omega} := \bigcup_{n} 2^n$ and let $2^{\leq \omega}:=  2^{<\omega} \cup 2^{ \omega}$. $X\in 2^{\omega}$ is said to be Martin-L{\"o}f random (MLR) if it is not in any effectively null set with respect to the uniform measure on $2^{\omega}$ \cite{misc}. Nies and Scholz introduced the notion of a \emph{state} to describe an infinite sequence of qubits \cite{unpublished} and then defined quantum-Martin-L{\"o}f randomness for states. We quickly sketch those parts of their work most relevant to the arguments in this paper. All definitions in the introduction section are from \cite{unpublished} or \cite{misc}. See \cite{misc} and \cite{misc1} for a detailed introduction to Martin-L{\"o}f randomness and computability theory.

\begin{defn}
A \emph{state}, $\rho=(\rho_n)_{n\in \mathbb{N}}$ is an infinite sequence of density matrices such that $\rho_{n} \in \mathbb{C}^{2^{n} \times 2^{n}}$ and $\forall n$,  $PT_{\mathbb{C}^{2}}(\rho_n)=\rho_{n-1}$. 
\end{defn}
Here, $PT_{\mathbb{C}^{2}}$ denotes the partial trace which `traces out' the last qubit from $\mathbb{C}^{2^n}$. $\rho$ is a infinite sequence of qubits whose first $n$ qubits are given by $\rho_n$. For this notion to be meaningful, $\rho$ is required to be coherent in the following sense; for all $n$, $\rho_n$, when `restricted' via the partial trace to it's first $n-1$ qubits, has the same measurement statistics as the state on $n-1$ qubits given by $\rho_{n-1}$. We now sketch a few notions from computability theory pertinent to us. Let $A\in 2^{\omega}$. We define an $A$-computable function to be a total function that can be realized by a Turing machine with $A$ as an oracle. By `computable', we will refer to $\emptyset$-computable. The concept of an $A$-computable sequence of natural numbers will come up frequently in our discussion.
\begin{defn}
A sequence $(a_n)_{n\in \mathbb{N}}$ is said to be $A$-computable if there is a $A$-computable function $\phi$ such that $\phi(n)=a_n$
\end{defn}

\begin{defn}
A special projection is a hermitian projection matrix with complex algebraic entries. 
\end{defn}
Since the algebraic complex numbers have a computable presentation (see \cite{unpublished}), we may identify a special projection with a natural number and hence talk about computable sequences of special projections. Let $I$ denote the two by two identity matrix.
\begin{defn}

A quantum $\Sigma_{1}^0$
set (or q-$\Sigma_{1}^0$
set for short) G is a computable
sequence of special projections $G=(p_{i})_{i\in \mathbb{N}}$ such that $p_i$ is $2^i$ by $2^i$ and range$(p_i \otimes I) \subseteq$ range $(p_{i+1}), \forall i\in \mathbb{N}$. 

\end{defn}

While a $2^n$ by $2^n$ special projection may be thought of as a computable projective measurement on a system of $n$ qubits, a q-$\Sigma_{1}^0$
class corresponds to a computable sequence of projective measurements on longer and longer systems of qubits and mirrors the concept of a $\Sigma_{1}^0$
class in computability theory. One of the many equivalent ways of defining a $\Sigma_{1}^0$
class is as follows. If $C \subset 2^n$, let $\llbracket C \rrbracket \subseteq  2^{\omega}$ be the set of all $X \in 2^{\omega}$ such that the initial segment of $X$ of length $n$ is in $C$.
\begin{defn}
\label{def:10}
A $\Sigma_{1}^0$ class $S \subseteq 2^{\omega}$ is any set of the form, \[S=\bigcup_{i\in \mathbb{N}} \llbracket A_{i} \rrbracket\] where
    \begin{enumerate}
        \item $A_{i} \subseteq 2^i, \forall i\in \mathbb{N}$ 
        \item
        The indices of $A_{i}$ form a computable sequence. (Being a finite set, each $A_i$ has a natural number coding it.
        \item $\llbracket A_{i} \rrbracket\subseteq \llbracket A_{i+1} \rrbracket, \forall i\in \mathbb{N}$
    \end{enumerate}
\end{defn}
A $\Sigma_{1}^0$ class, S is coded (non-uniquely) by the index of the total computable function generating the sequence $(A_{i})_{ i\in \mathbb{N}}$ occurring in (2) in the definition of $S$. Hence, the notion of a computable sequence of $\Sigma_{1}^0$ classes makes sense (see \cite{misc}, section 3.2). One sees that the special projections $q_i$ in the definition of the q-$\Sigma_{1}^0$, $G$ play the role of the $A_i$s which generate a the $\Sigma_{1}^0$ class, $S$. The following notion is a quantum analog of the uniform measure of $S$ which equals $\lim_{n}(2^{-n}|A_n|)$, where $|.|$ refers to the cardinality. (The uniform measure on $2^{\omega}$ is the measure induced by letting the measure of $ \llbracket \tau \rrbracket$ to be $2^{-|\tau|}$ for each $\tau \in 2^{<\omega}$. Here, $|\tau|:=n$ if $\tau \in 2^n$.)
\begin{defn}

If $G=(p_{n})_{n\in \mathbb{N}}$ is a q-$\Sigma_{1}^0$ class, define $\tau(G):=\lim_{n}(2^{-n}|q_n|)$ where, $|q_n|$ is the rank of $q_n$.
\end{defn}
 Informally, a q-$\Sigma_{1}^0$ class,  $G=(p_{n})_{n\in \mathbb{N}}$ may be thought of as an observable whose expected value, when `measured' on a state $\rho=(\rho_{n})_{n\in \mathbb{N}}$ is $\rho(G):=\lim_{n}$ tr $(\rho_{n}p_n)$. The positive-semidefiniteness of density matrices, the coherence of the components of $\rho$ and condition (3) in definition 1.5 ensure that $ tr(\rho_{n}p_n) $ is non decreasing in $n$. As   $tr(\rho_{n}p_n)\leq 1$ for all $n$, $\rho(G):=\lim_{n}$ tr $(\rho_{n}p_n)$ exists.

\begin{defn}
A classical Martin-L{\"o}f test (MLT) is a computable sequence, $(S_{m})_{m \in \mathbb{N} }$ of $\Sigma_{0}^{1}$ classes such that the uniform measure of $S_m$ is less than or equal to $2^{-m}$ for all m.
\end{defn}
\begin{defn}
A quantum Martin-L{\"o}f test (q-MLT) is a computable sequence, $(S_{m})_{m \in \mathbb{N}}$ of q-$\Sigma_{0}^{1}$ classes such that $\tau(S_m)$ is less than or equal to $2^{-m}$ for all m.
\end{defn}
Having established the necessary prerequisites, we can define a quantum Martin-L{\"o}f random (q-MLR) state. Roughly speaking, a state is q-MLR if it cannot be `detected by projective measurements of arbitrarily small rank'.
\begin{defn}
$\rho$ is q-MLR if for any q-MLT $(S_{m})_{m \in \mathbb{N}}$, inf$_{m \in \mathbb{N}}\rho(S_m)=0$.
\end{defn}
\begin{defn}
$\rho$ is said to fail the q-MLT $(S_{m})_{m \in \mathbb{N}}$, at order $\delta$, if inf$_{m \in \mathbb{N}}\rho(S_m)>\delta$.
\end{defn}

\section{Measuring a state induces a measure on $2^{\omega}$}
To fix notation, let $X(n)$ denote the $n$th bit of an $X \in 2^{\leq \omega}$ , let $p(E)$ stand for the probability of the event $E$.
\begin{defn}
\label{defn:78}
An A-computable measurement system $B= ((b^{n}_{0},b^{n}_{1}))_{n\in \mathbb{N}}$ (or just `measurement system' for short) is a sequence of orthonormal bases for $\mathbb{C}^{2}$ such that each $b^{n}_{i}$ is complex algebraic and the sequence $ ((b^{n}_{0},b^{n}_{1}))_{n\in \mathbb{N}}$ is A-computable.
\end{defn}

Let $\rho=(\rho_n)_{n\in \mathbb{N}}$ be a state and $B= ((b^{n}_{0},b^{n}_{1}))_{n\in \mathbb{N}}$ be a measurement system. We now work towards formalizing a notion of \emph{qubitwise} measurement of $\rho$ in the bases in $B$.
 A (probability) premeasure \cite{misc1},$p$ (also called a measure representation \cite{misc}), is a function from the set of all finite bit strings to $[0,1]$ satisfying  $\forall n$, $\forall \tau \in 2^n, p(\tau)=p(\tau 0)+p(\tau 1)$. $p$ induces a measure on $2^{\omega}$ which is seen to be unique by Carath\'eodory's extension theorem (See 6.12.1 in \cite{misc1}). Flipping a $0,1$ sided fair coin repeatedly induces a probability measure (which happens to be the uniform measure) on $2^{\omega}$ as follows. Let the random variable $Z(n)$ denote the outcome of the the $n$th coin flip. The sequence $(Z(n))_{n\in \mathbb{N}}$ induces a premeasure, $p$, on $ 2^{<\omega}$ which extends to the uniform measure on $2^{\omega}$. Here, $p(\sigma)=2^{-n}$ is the probability that $Z(i)=\sigma(i)$ for all $i\leq |\sigma|$. Similarly the act of measuring $\rho$ qubit by qubit in $B$ induces a premeasure on $2^{<\omega}$ which extends to a probability measure (denoted $\mu_{\rho}^{B}$) on $2^{\omega}$ as follows. Let the random variable $X(n)$ be the $0,1$ valued outcome of the measurement of the $n$th qubit of $\rho$. Let $p$ be the premeasure induced by the sequence $(X(n))_{n \in \mathbb{N}}$ on $ 2^{<\omega}$. $p$ extends to $\mu_{\rho}^{B}$ on $2^{\omega}$. For any $A \subseteq 2^{\omega}$, $\mu_{\rho}^{B}(A)$ is the probability that $X \in A$ where $X$ is the element of $2^{\omega}$ obtained in the limit by the qubit by qubit measurement of $\rho$ in $B$. The most conspicuous difference between the two situations is that while the $(Z(n))_{n\in \mathbb{N}}$ are independent, $(X(n))_{n\in \mathbb{N}}$ need not be independent as the elements of $\rho$ can be entangled. We now formalize the above. The following calculations follow from standard results mentioned, for example, in \cite{bookA}. 
 
 We now define $(X(n))_{n \in \mathbb{N}}$ and $p$, the induced  premeasure. Measure $\rho_1$ by the measurement operators $\{|b^{1}_0\big>\big<b^{1}_0|,|b^{1}_1\big>\big<b^{1}_1|\}$ and define $X(1) := i$ if $b^{1}_{i}$ was obtained by the above measurement. Let $\hat{\rho_{2}}$ be the density matrix corresponding to the post-measurement state of $\rho_2$ given that $\rho_2$ yields $|b^{1}_{X(1)}\big> \big<b^{1}_{X(1)}| \otimes I$ if measured in the system 
 \[ (|b^{1}_{i}\big> \big<b^{1}_{i}|   \otimes I )_{i\in\{0,1\}}.\]
  
 I.e,
\[\hat{\rho_{2}} = \dfrac{(|b^{1}_{X(1)}\big> \big<b^{1}_{X(1)}| \otimes I) \rho_2 (|b^{1}_{X(1)}\big> \big<b^{1}_{X(1)}| \otimes I)}{tr((|b^{1}_{X(1)}\big> \big<b^{1}_{X(1)}| \otimes I) \rho_2 \big)}. \]

 To define $X(2)$, measure $\hat{\rho_{2}}$ by the measurement operators 
\[ (I \otimes |b^{2}_{i}\big> \big<b^{2}_{i}|)_{i\in\{0,1\}},\]
and set $X(2):= i$ if $I \otimes |b^{2}_{i}\big> \big<b^{2}_{i}|$ is obtained. We use $\hat{\rho_{2}}$ instead of $\rho_2$ to define $X(2)$ to account for the previous measurement of the first qubit. $X(n)$ is defined similarly.
By the above, 
 \[p(ij):=p(X(1)=i,X(2)=j) = p(X(1)=i)p(X(2)=j|X(1)=i)=\]\[p(X(1)=i) tr \big[I \otimes |b^{2}_{j}\big> \big<b^{2}_{j}| (\dfrac{(|b^{1}_{i}\big> \big<b^{1}_{i}| \otimes I) \rho_2 (|b^{1}_{i}\big> \big<b^{1}_{i}| \otimes I)}{tr((|b^{1}_{i}\big> \big<b^{1}_{i}| \otimes I) \rho_2)})\big].\]
 
 Since $PT_{\mathbb{C}^{2}}(\rho_2)=\rho_{1}$, $p(X(1)=i)=tr((|b^{1}_{i}\big> \big<b^{1}_{i}| \otimes I) \rho_2\big)$. So, \[p(ij)=tr \big[ \rho_2 (|b^{1}_{i}b^{2}_{j}\big> \big<b^{1}_{i}b^{2}_{j}|)\big].\]
 
 Given $\tau \in 2^n$, similar calculations show that 
 \begin{align}
 \label{eq:21}
    p(\tau) := p(X(1)=\tau(1), \dots,X(n)=\tau(n)) = tr \big[ \rho_n (|\bigotimes_{i=1}^{n} b^{i}_{\tau(i)}\big> \big<\bigotimes_{i=1}^{n} b^{i}_{\tau(i)}\big| \big]. 
 \end{align}

 This defines $p$. The following lemma shows that $p(.)$ is a premeasure. Define $\mu^{B}_{\rho}$ to be the unique probability measure induced by it.
 \begin{lem}
 $\forall n$, $\forall \tau \in 2^n, p(\tau)=p(\tau 0)+p(\tau 1)$
 \end{lem}
 \begin{proof}
Noting that for $j\in \{0,1\}$,
  \[
    \rho_{n+1} (|\bigotimes_{i=1}^{n} b^{i}_{\tau(i)} \otimes b^{n+1}_{j}\big> \big<\bigotimes_{i=1}^{n} b^{i}_{\tau(i)} \otimes b^{n+1}_{j}|) =\rho_{n+1} (|\bigotimes_{i=1}^{n} b^{i}_{\tau(i)}\big> \big<\bigotimes_{i=1}^{n} b^{i}_{\tau(i)}|\otimes |b^{n+1}_{j}\big>\big<b^{n+1}_{j}|),\]
    
    and letting $A:=|\bigotimes_{i=1}^{n} b^{i}_{\tau(i)}\big> \big<\bigotimes_{i=1}^{n} b^{i}_{\tau(i)}| $, the right hand side is
     \[=
    tr\big[(A\otimes |b^{n+1}_{0}\big>\big<b^{n+1}_{0}|)\rho_{n+1} + (A\otimes |b^{n+1}_{1}\big>\big<b^{n+1}_{1}|)\rho_{n+1} ]  \]
     \[=
    tr\big[(A\otimes (|b^{n+1}_{0}\big>\big<b^{n+1}_{0}|  + |b^{n+1}_{1}\big>\big<b^{n+1}_{1}|))\rho_{n+1} ]=
    tr\big[(A\otimes I) \rho_{n+1} ] = tr[A \rho_n] = p(\tau)\]
\end{proof}
\begin{remark}
\label{rem:1}
 If $B$ is $S$-computable and $\rho$ is $T$-computable, then the sequence $\{\mu^{B}_{\rho}(\sigma)\}_{\sigma \in \mathbb{N}}$ is $S\oplus T$-computable.
\end{remark} 
Here, $S\oplus T$ is obtained by putting $S$ on the even bits and $T$ on the odd bits \cite{misc}.

\section{Measurement Randomness}
Let $MLR\subset 2^\omega$ be the set of MLR bitstrings. If $\rho$ is a state and $B$ a measurement system, $\mu^{B}_{\rho} (MLR)$ is the probability of getting a MLR bitstring by a qubit-wise measurement of $\rho$ as described in the previous section. 
\begin{defn}
$\rho$ is measurement random (mR) if for any computable measurement system, B, $\mu^{B}_{\rho} (MLR)=1$ 
\end{defn}
\begin{thm}
All q-MLR states are also mR states.
\end{thm}
\begin{proof}

 Let $\rho=(\rho_n)_{n\in \mathbb{N}}$ be q-MLR. Suppose towards a contradiction that there is a $\delta \in (0,1)$ and a computable $B= ((b^{n}_{0},b^{n}_{1}))_{n\in \mathbb{N}}$ such that $\mu^{B}_{\rho} (2^{\omega}/MLR)>\delta$. Let $(S^{m})_m$ be the universal MLT \cite{misc} and let for all $m$,
\begin{align}
\label{eq:20}
   S^{m}=\bigcup_{m\leq i} \llbracket A^{m}_{i} \rrbracket, 
\end{align}

where the $A^{m}_{i}$s satisfy the conditions of Definition \ref{def:10}. By the definition of a MLT, for all $m$ and all $i\geq m$, we can write $A^{m}_{i}= \{\tau^{m,i}_{1},\dots,\tau^{m,i}_{k^{m,i}}\} \subset 2^i$ for some $0 \leq k^{m,i} \leq 2^{i-m}$. Now define a q-MLT as follows. For all $m$ and $i\geq m$, let $\tau_{a} = \tau_{a}^{m,i}$ for convenience and define the special projection:
\begin{align}
\label{eq:23}
   p^{m}_{i}= \sum_{a\leq k^{m,i}} (|\bigotimes_{q=1}^{i} b^{q}_{\tau_{a}(q)}\big> \big<\bigotimes_{q=1}^{i} b^{q}_{\tau_{a}(q)}\big|\big).
\end{align}
Letting $P^{m}:=(p^{m}_{i})_{m\leq i}$, we see that $(P^{m})_{m \in \mathbb{N}}$ is a q-MLT (For each $m$, the sequence $(p^{m}_{i})_{m\leq i}$ is computable since $B$ and $(A^{m}_{i})_{m\leq i}$ are computable. Condition 3 in Definition \ref{def:10} implies that for all $i$, range$(p^{m}_{i})\subseteq$range$(p^{m}_{i+1})$. So, $P^{m}$ is a q-$\Sigma_{0}^{1}$ class for all $m$. $k^{m,i} \leq 2^{i-m}$ for all $m,i$ implies that $\tau(P^m)\leq 2^{-m}$ for all $m$. Since $(S^{m})_{m\in \mathbb{N}}$ is a MLT, $(P^{m})_{m\in \mathbb{N}}$ is a computable sequence.)
For all m, $ (2^{\omega}/MLR) \subseteq S^{m}$ holds by the definition of a universal MLT. Hence, since \ref{eq:20} is an increasing union and as $\mu^{B}_{\rho}(2^{\omega}/MLR)>\delta$, for all $m$ there exists an $i(m)>m$ such that
\begin{align}
\label{eq:22}
\mu^{B}_{\rho}(\llbracket A^{m}_{i(m)} \rrbracket)> \delta.
\end{align} Fix such an $m$ and corresponding $i=i(m)$ and let $A^{m}_{i}= \{\tau_{1},\dots,\tau_{k^{m,i}}\}$ for some $k^{m,i} \leq 2^{i-m}$ as in \ref{eq:23}. By \ref{eq:21} and \ref{eq:22}, we have that
\begin{align}
\label{eq:24}
\delta<\sum_{a\leq k}   p(\tau_{a}) = \sum_{a\leq k^{m,i}} tr \big[ \rho_i (|\bigotimes_{q=1}^{i} b^{q}_{\tau_{a}(q)}\big> \big<\bigotimes_{q=1}^{i} b^{q}_{\tau_{a}(q)}\big|\big)\big]=tr \big[ \rho_i\sum_{a\leq k^{m,i}} (|\bigotimes_{q=1}^{i} b^{q}_{\tau_{a}(q)}\big> \big<\bigotimes_{q=1}^{i} b^{q}_{\tau_{a}(q)}\big|\big)\big]
\end{align}
So, by \ref{eq:23} and \ref{eq:24}, we see that for all $m$ there is an $i$ such that,
\[\delta < tr[\rho_{i}p^{m}_{i}] \leq \rho(P^m).\]
So, inf$_{m}(\rho (P^{m}))>\delta$, contradicting that $\rho$ is q-MLR.
\end{proof}

\begin{defn}
$\rho=(\rho_{n})_{n \in \mathbb{N}}$ is computable if the sequence $(\rho_{n})_{n \in \mathbb{N}}$ is computable.
\end{defn}

\begin{thm}
\label{thm:0000}
There is a computable state which is not q-MLR but is mR.
\end{thm}
\begin{proof}
All matrices in this proof are in the standard basis.
Let $\rho= \bigotimes_{n=5}^{\infty} d_{n}$ and for $N>5$, $S_{N}:= \bigotimes_{n=5}^{N} d_{n}$.
where $d_n$ is a $2^n$ by $2^n$ matrix with $2^{-n}$ along the diagonal and $r_{n}:=\lfloor 2^n/n \rfloor $ many $2^{-n}$s on the extreme ends of the anti-diagonal. Formally, define $d_n$ to be the symmetric matrix such that:
For $i \leq r_n$, $d_{n}(i,j)= 2^{-n}$ if $j=i$ or $j=2^{n}-i+1$ and $d_{n}(i,j)=0$ otherwise.
For $ r_n < i < 2^{n}-r_{n}$, $d_{n}(i,j)= 2^{-n}$ if $j=i$ and $d_{n}(i,j)=0$ otherwise. For example, $r_{3}=2$ and so, 
\[   d_3 = 
\begin{bmatrix}
2^{-3} & 0 & 0 & 0 & 0 & 0 & 0 & 2^{-3}\\
0 & 2^{-3} & 0 & 0 & 0 & 0 & 2^{-3}  & 0\\
0 & 0 & 2^{-3} & 0 & 0 & 0 & 0 & 0\\
0 & 0 & 0 & 2^{-3}  & 0 & 0 & 0 & 0\\
0 & 0 & 0 &  0 & 2^{-3} & 0 & 0 & 0\\
0 & 0 & 0 & 0 & 0 & 2^{-3} & 0 & 0\\
0 & 2^{-3}  & 0 & 0 & 0 & 0 & 2^{-3}  & 0\\
2^{-3} & 0 & 0 & 0 & 0 & 0 & 0 & 2^{-3}\\
\end{bmatrix}
\]
Clearly, $d_n$ is a density matrix.
The theorem will be proved via the following lemmas.
\begin{lem}
\label{lem:20}
$\rho$ is not q-MLR.
\end{lem}
\begin{proof} It is easy to see that zero has multiplicity $r_n$ as an eigenvalue of $d_n$. Hence, letting $q_{n}= 2^{n}-r_{n}$, the eigenpairs of $d_n$ can be listed as $\{\alpha^{n}_{i}, v^{n}_i\}_{i=1}^{2^{n}}$ where $\alpha^{n}_{i}=0$ if $ q_{n}+1 \leq i \leq 2^{n} $ and $(v^{n}_i)_{i=1}^{2^{n}}$ is a orthonormal basis of $\mathbb{C}^{2^{n}}$.

Fix a $N>5$. By properties of the Kronecker product, $S_{N}$ has a orthonormal basis of eigenvectors:  \[\{ \bigotimes_{n=5}^{N} v^{n}_{l(n)}: (l(n))_{n=5}^{N} \text{ is a sequence such that for all } n, l(n) \leq 2^n  \},\]
and $\bigotimes_{n=5}^{N} v^{n}_{l(n)}$ has eigenvalue $ \prod_{n=5}^{N} \alpha^{n}_{l(n)} $.
Letting $M_N$ be those elements of the above eigenbasis having non-zero eigenvalues, we have that

\[
\label{eq:eig}
M_{N}=\{ \bigotimes_{n=5}^{N} v^{n}_{l(n)}: (l(n))_{n=5}^{N} \text{ is a sequence such that for all } n,  l(n) \leq q_{n}  \}.\]
By the definition of $q_n$, \[|M_N|= \prod_{n=5}^{N} 2^{n}-\lfloor 2^n/n \rfloor \leq \prod_{n=5}^{N} 2^{n}- (2^n/n) + 1 = \prod_{n=5}^{N} 2^{n}(1 - n^{-1} + 2^{-n}) = \prod_{n=5}^{N} 2^{n} \prod_{n=5}^{N}  (1 - n^{-1} + 2^{-n}).\]
Noting that $\prod_{n=5}^{\infty}  (1 - n^{-1} + 2^{-n}) =0 $,  define a q-MLT $(T_{m})_{m\in \mathbb{N}}$ as follows. Given $m$, we describe the construction of $T_{m}$. Find $N=N(m)$ such that
$\prod_{n=5}^{N}  (1 - n^{-1} + 2^{-n})< 2^{-m}$. Let $\gamma(N):= \sum_{n=5}^{N} n$  and let \[p_{\gamma(N)}=   \sum_{v \in M_{N}}   |v \big>\big<v |.\]
$p_{\gamma(N)}$ is a special projection on $\mathbb{C}^{2^{{\gamma(N)}}}$ having rank equal to $|M_{N}|$. Let $p_{k}=\emptyset$ for $k<\gamma(N)$ and \[p_{k}:= p_{\gamma(N)} \otimes \bigotimes_{i=1} ^{k-\gamma(N)} I   \]
for $k>\gamma(N)$. Using that $\rho$ is computable, it is easy to see that $(p_{k})_{k\in \mathbb{N}}$ is a q-$\Sigma^{0}_{1}$ class. Let $T_{m}:= (p_{k})_{k\in \mathbb{N}}$. $(T_{m})_{m\in \mathbb{N}}$ is a q-MLT since the choice of $N(m)$ implies that $\tau(T_{m}) < 2^{-m}$ and as $N(m)$ can be computed from $m$. $(T_{m})_{m\in \mathbb{N}}$ demonstrates that $\rho$ is not q-MLR as follows. Fix $m$ arbitrarily and let $N(m)$ be as above. Recalling that $M_{N}$ is the set consisting of all eigenvectors of $S_{N}$ with non-zero eigenvalue, we have that, \[\rho(T_{m}) \geq tr(\rho_{\gamma(N)}p_{\gamma(N)}) =  tr(S_{N}p_{\gamma(N)})=tr(S_{N})=1.\] Since $m$ was arbitrary, $inf_{m \in \mathbb{N}}(\rho(T_{m}))=1$.
\end{proof}
The following technical lemma, although seems unmotivated at this juncture, is crucial at a later point in the proof.  
\begin{lem}
\label{lem:21}
Let $\{[a_{i}, b_{i}]^{T}\}_{i=1}^{n} $ be a set of unit column vectors in $\mathbb{C}^2$. Let $V=\bigotimes_{i=1}^{n}[a_{i}, b_{i}]^{T}$ be their Kronecker product. If  $V=[v_{1},v_{2},\dots,v_{2^n}]^{T}$, then for all $k \leq 2^{n-1}$, we have that \[|v_{k}||v_{2^{n}-k+1}| = \prod_{i=1}^{n} |a_{i}||b_{i}|.\]
\end{lem}
\begin{proof}

For natural numbers $u$ and $q$, let $[u]_{q}$ denote the remainder obtained by dividing $u$ by $q$. We use the following convention for the Kronecker product \cite{Regalia:1989:KPU:76594.76599}:

\[\begin{bmatrix}
a_1\\
b_1\\
\end{bmatrix}
\otimes
\begin{bmatrix}
a_2\\
b_2\\
\end{bmatrix}
=
\begin{bmatrix}
a_1 a_2\\
b_1 a_2\\
a_1 b_2\\
b_1 b_2
\end{bmatrix}
.\]
So, $v_{1}= \prod_{i=1}^{n} a_{i}$ and $v_{2^{n}}= \prod_{i=1}^{n} b_{i}$. For any $k\leq 2^{n-1}$, $v_k$ has the form $v_k= \prod_{i=1}^{n} c^{k}_{i}$, for some  $c^{k}_{i} \in \{a_{i}, b_{i}\}$ and $v_{2^{n}-k+1}$ has the form $v_{2^{n}-k+1} = \prod_{i=1}^{n} e^{k}_{i}$, for some  $e^{k}_{i} \in \{a_{i}, b_{i}\}$. Note that $c^{k}_{1}=a_{1}$ if and only if $k$ is odd if and only if   $e^{k}_{1}=b_{1}$. Similarly, we have the following. $c^{k}_{2}=a_{2}$ if and only if $[k]_{2^{2}} \in \{1,2\}$ if and only if $e^{k}_{2}=b_{2}$.
$c^{k}_{3}=a_{3}$ if and only if $[k]_{2^{3}} \in \{1,\dots,2^{2}\}$ if and only if $e^{k}_{3}=b_{3}$.
In general, for $i\leq n$, for all $k\leq 2^{n-1}$, \[c^{k}_{i}=a_{i} \iff  [k]_{2^{i}} \in \{1,\dots,2^{i-1}\} \iff e^{k}_{i}=b_{i}.\]
This proves the lemma. Intuitively, this happens for the following reason. Imagine moving from $v_1$ to $v_{2^{n-1}}$ (by incrementing $k$) and keeping track of the values of $c^{k}_i$ as you move along the $v_k$s. Also, imagine moving from $v_{2^n}$ to $v_{2^{n-1}}$ and keeping track of the values of $e^{k}_i$ as you move along the $v_{2^n - k +1}$s. Both motions are in opposite directions since as $k$ is incremented, the first motion is from lower to higher indices and the second is from higher to lower indices. Consider the behavior of $c^{k}_1,e^{k}_1$ as $k$ is incremented. At the `start' point, $c^{1}_1=a_1$, $e^{1}_1=b_1$. Now, as you move (i.e as you increment $k$), $c^{k}_1$ alternates between $a_1$ and $b_1$ equalling it's starting value, $a_1$ at odd $k$s and $e^{k}_1$ alternates between $b_1$ and $a_1$ equalling it's starting value $b_1$  for odd $k$s. Now, take any $i\leq n$.
$c^{k}_i$ alternates between $a_i$ and $b_i$ in blocks of length $2^{i-1}$. $c^{k}_i=a_i$ when $k$ is in the first block, $\{1,2,\dots, 2^{i-1}\}$ (i.e, when $[k]_{2^{i}} \in \{1,2,\dots, 2^{i-1}\}$) and $c^{k}_i=b_i$ when $k$ is in the second block, $\{2^{i-1}+1,\dots, 2^{i}\}$(i.e, when $[k]_{2^{i}} \in \{2^{i-1}+1,\dots, 0\}$) and so on. Similarly, $e^{k}_i$ alternates between $b_i$ and $a_i$ in blocks of length $2^{i-1}$.
\end{proof}

\begin{lem}
\label{lem:22}
Let $n\in \mathbb{N}$ and let $\{[a_{i}, b_{i}]^{T}\}_{i=1}^{n} $ be such that for all i, $[a_{i}, b_{i}]^{T}$ is unit column vector in $\mathbb{C}^2$ and let $W=\bigotimes_{i=1}^{n}[a_{i}, b_{i}]^{T}$. Then, $|\big<W|d_{n}|W\big>| \in [2^{-n}(1-2n^{-1}),2^{-n}(1+2n^{-1})]$
\end{lem}
\begin{proof}

Fix $n$ and $V$ as in the statement and write $d_n$ as a block matrix with each block of size $2^{n-1}$ by $2^{n-1}$.
\[d_n = \begin{bmatrix}
A & B\\
B^{T} & A\\
\end{bmatrix}.
\]
Letting $V=\bigotimes_{i=1}^{n-1}[a_{i}, b_{i}]^{T}$, in block form, $W=[a_{n}V^{T}, b_{n}V^{T}]^{T} $. Let $V=[v_{1},v_{2},\dots,v_{2^{n-1}}]^{T}$. It is easily checked that 
\[\big<W|d_{n}|W\big>= 2^{-n} +  a_{n}^{*} b_{n} V^{\dagger}BV+ a_{n}b_{n}^{*} V^{\dagger}B^{T}V.\]
By the form of B we get, \[V^{\dagger}BV = 2^{-n} [v^{*}_{1},v^{*}_{2},\dots,v^{*}_{2^{n-1}}][v_{2^{n-1}},v_{2^{n-1}-1},\dots,v_{2^{n-1}-r_{n}+1},0,\dots ,0]^{T}. \]

\[= 2^{-n}\sum_{k=1}^{r_{n}}v^{*}_{k}v_{2^{n-1}-k+1}.\]
By the previous lemma, \[|V^{\dagger}BV| \leq 2^{-n}\sum_{k=1}^{r_{n}}|v_{k}||v_{2^{n-1}-k+1}| = 2^{-n} r_{n} \prod_{i=1}^{n-1} |a_{i}||b_{i}| = 2^{-n} r_{n} \prod_{i=1}^{n-1} |a_{i}|\sqrt{1-|a_{i}|^{2}}. \]
Since $x\sqrt{1-x^{2}}$ has a maximum value of $1/2$ and recalling definition of $r_n$,
\[|V^{\dagger}BV| \leq 2^{-n}\dfrac{1}{2^{n-1}}\dfrac{2^{n}}{n} = \dfrac{2^{1-n}}{n}. \]
Similarly, 
$|V^{\dagger}B^{T}V| \leq  \dfrac{2^{1-n}}{n}.$ Noting that $|a_{n}^{*} b_{n}|, |a_{n} b_{n}^{*}| \leq 1/2$, \[ |\big<W|d_{n}|W\big>| \leq 2^{-n} +  |a_{n}^{*} b_{n} V^{\dagger}BV| + |a_{n}b_{n}^{*} V^{\dagger}B^{T}V|\leq 2^{-n} +  \dfrac{2^{1-n}}{n},\]
and 
\[ |\big<W|d_{n}|W\big>| \geq 2^{-n} - |a_{n}^{*} b_{n} V^{\dagger}BV| - |a_{n}b_{n}^{*} V^{\dagger}B^{T}V|\geq 2^{-n} -  \dfrac{2^{1-n}}{n}.\]
\end{proof}
\begin{lem}
\label{lem:24}
$\rho$ is mR.
\end{lem}
If $p$ is any measure on $2^{\omega}$, we can define Martin-L{\"o}f randomness with respect to $p$ exactly as we defined it for the uniform measure. Denote by $MLR(p)$, the set of bitstrings Martin-L{\"o}f random with respect to $p$ \cite{bookE}.\\
\begin{proof}
We use ideas similar to Theorem 196(a) in \cite{bookE}. For convenience, for all $i>5$, define \[\beta_{i} := \sum_{q=5}^{i-1}q.\]
Let $B$ be any computable measurement system. We show that $ MLR(\mu^{B}_{\rho}) \subseteq MLR$. Since $\mu^{B}_{\rho}[MLR(\mu^{B}_{\rho})]=1$, this implies that $\mu^{B}_{\rho}(MLR)=1$. Denote $\mu^{B}_{\rho}$ by $\mu$ for convenience. Let $\lambda$ denote the uniform measure.  We will abuse notation by writing $\mu(\tau)$ instead of the more cumbersome $\mu(\llbracket\tau\rrbracket)$ for $\tau \in 2^{<\omega}$.
Let $X \in MLR(\mu)$. Write $X$ as a concatenation of finite bitstrings : $X=\sigma_{5}\sigma_{6}\dots \sigma_{n}\dots$ where $\sigma_{n} \in 2^n$ for all $n\in \mathbb{N}$. Let $S_n := \sigma_{5}\sigma_{6}\dots \sigma_{n}$ be the concatenation upto $n$. Let $\mu_{i}$ be such that for all $\tau \in 2^i$, \[\mu_{i}(\tau):= tr\big[d_{i}(|\bigotimes_{q=1}^{i} b^{q+\beta_{i}}_{\tau(q)}\big> \big<\bigotimes_{q=1}^{i} b^{q+\beta_{i}}_{\tau(q)}|)\big].\]
By \ref{eq:21} and by the form of $\rho$ we see that,
\[\mu(S_{n})= \prod_{i=5}^{n} \mu_{i}(\sigma_{i}).\]
Note that $\mu$ is computable \cite{bookE} since $\rho$ and $B$ are. Since $X \in MLR(\mu)$, by the Levin-Schnorr theorem (Theorem 90, section 5.6 in \cite{bookE}) there is a $C_{1}$ such that
\begin{align*}
    \forall n ,  -\log(\mu(S_{n})) - C_{1} \leq KM(S_{n}). 
\end{align*}
By Theorem 89, section 5.6 in \cite{bookE} fix a $C_2$ such that
\begin{align*}
    \forall n ,   KM(S_{n}) \leq -\log(\lambda(S_{n})) + C_{2}.
\end{align*} 
By these inequalities and taking exponents, we see that there is a constant $\alpha>0$ such that \[\forall n,  \mu(S_{n}) \geq \alpha \lambda(S_{n}).\]
Letting $r_{i}:= \mu_{i}(\sigma_{i})$ and $\delta_{i}:= \lambda(\sigma_{i})-r_{i}$ in the above,
\begin{align}
\label{eq:25}
    \forall n,  \prod_{i=5}^{n} r_{i}  \geq \alpha \prod_{i=5}^{n} r_{i}+ \delta_{i}.
\end{align}
Let $\mu'$ be a probability measure on $2^{\omega}$ such that for all $\sigma\in 2^{<\omega}, \mu'(\sigma):=2\mu(\sigma)-\lambda(\sigma)$. In particular, this implies that \[\forall n,\mu'(S_{n}) = \prod_{i=5}^{n} r_{i}-\delta_{i}.\]
Note that $\mu'$ is computable since $\mu$ and $\lambda$ are. Applying the same argument which resulted in \ref{eq:25}, we get that there is an $\epsilon>0$ such that,
\begin{align}
\label{eq:26}
    \forall n,  \prod_{i=5}^{n} r_{i}  \geq \epsilon \prod_{i=5}^{n} r_{i}- \delta_{i}.
\end{align}
By Lemma \ref{lem:22}, for all $i, r_{i}\in [2^{-i}(1-2i^{-1}),2^{-i}(1+2i^{-1})]. $
So, $|\delta_{i}|=|r_{i}-2^{-i}| \in [0,2^{-i+1}i^{-1}]$.
Hence, \\$r_{i}+\delta_{i} \geq 2^{-i}-2^{-i+1}i^{-1} - 2^{-i+1}i^{-1} = 2^{-i}[1-4i^{-1}]>0$, since $i\geq 5$. Similarly, $r_{i}-\delta_{i}\geq 0$. By this, multiplying \ref{eq:25} and \ref{eq:26} gives,
\begin{align}
\label{eq:27}
    \forall n,  \prod_{i=5}^{n} r^{2}_{i}  \geq \alpha \epsilon \prod_{i=5}^{n} r^{2}_{i}- \delta^{2}_{i}=\alpha \epsilon\prod_{i=5}^{n} r^{2}_{i} \prod_{i=5}^{n} \big(1-\dfrac{\delta^{2}_{i}}{r^{2}_{i}}\big).
\end{align}
By the above, \[\dfrac{|\delta_{i}|}{r_{i}} \leq \dfrac{2^{-i+1}i^{-1}}{2^{-i}(1-2i^{-1})} = 2(i-2)^{-1}.\] 
Letting $F>0$ be the constant,
\begin{align*}
    \forall n,  \prod_{i=5}^{n} \big(1-\dfrac{\delta^{2}_{i}}{r^{2}_{i}}\big) \geq \prod_{i=5}^{\infty} \big(1-\dfrac{\delta^{2}_{i}}{r^{2}_{i}}\big) \geq \prod_{i=5}^{\infty} \big(1-4(i-2)^{-2}\big)=F,
\end{align*}
\ref{eq:27} gives,
\begin{align}
\label{eq:28}
    \forall n,  (\alpha \epsilon)^{-1}\prod_{i=5}^{n} r^{2}_{i}  \geq  \prod_{i=5}^{n} r^{2}_{i}- \delta^{2}_{i} \geq \prod_{i=5}^{n} r^{2}_{i} F.
\end{align}
From \ref{eq:25}, \ref{eq:26} and \ref{eq:28}, it is easy to see that there is a $G>0$ such that for all $n$
\[\prod_{i=5}^{n} r_{i}+ \delta_{i} \geq G \prod_{i=5}^{n} r_{i}. \]
Recalling the definitions of $r_i$ and $\delta_i$, 
\[\forall n, \lambda(S_n) \geq G \mu(S_n).  \]
Letting $D= C_{1} -\log(G)$ and recalling the definition of $C_1$,
\[\forall n, -\log(\lambda(S_n)) \leq  -\log(\mu(S_n)) -\log(G) \leq  KM(S_n) + D. \]
By Theorem 85 in \cite{bookE}, $KM(.) \leq K(.) + O(1)$ and so there is a $E>0$ such that
\[\forall n, -\log(\lambda(S_n))  \leq  K(S_n) + E. \]
Noting that $-\log(\lambda(S_n))=|S_n|= \beta_{n}+n$, 3.2.14 from \cite{misc} implies that $X$ is MLR.
\end{proof}
The theorem is proved. 
\end{proof}
\section{Generalizations and future directions}
We sketch some ways in which the previous section's results generalize. Given $S\in 2^{\omega}$, we may relativize the notion of Martin-L{\"o}f randomness to define the set $MLR^{S} \subset 2^{\omega}$ of infinite bitstrings which are Martin-L{\"o}f random with respect to $S$. The halting problem, denoted by $\emptyset^{\prime} \subset \mathbb{N}$ is an incomputable set important in computability theory. Letting $\emptyset^{(n)}$ be the $n-1$th iterate of the halting problem, an element of Cantor space is said to be \emph{arithmetically random} if it is in $MLR^{\emptyset^{(n)}}$ for every $n$ (see 6.8.4 in \cite{misc1}). Given 
$ S \in 2^{\omega} $, relativizing the proof of Lemma \ref{lem:24} shows that $MLR^{S}(\mu^{B}_{\rho}) \subseteq MLR^{S}$ as follows. Take an $X\in MLR^{S}(\mu^{B}_{\rho})$. Relativizing Theorems 85 and 90 from \cite{bookE} and 3.2.14 from \cite{misc} to $S$ and noting that $KM^{S}(.) \leq KM(.)$ and following the proof of Lemma \ref{lem:24}  shows that 
$X \in MLR^{S}$. This shows that $\mu^{B}_{\rho}(MLR^{S})=1$ holds for any $S \in 2^{\omega}$ and any computable measurement system $B$. In particular, this has an interesting application; if $B$ is any computable measurement system, 
$\mu^{B}_{\rho}(MLR^{\emptyset^{(n)}})=1$ for all $n$. So, \[\mu^{B}_{\rho}\big[\bigcap_{n\in \mathbb{N}}(MLR^{\emptyset^{(n)}})\big]=1.\]
So, measuring $\rho$ in any computable measurement system yields an arithmetically random  infinite sequence of bits, with probability one. The above note naturally suggests a definition:

\begin{defn}
$\rho$ is said to be strong measurement random (strong mR), if $\mu^{B}_{\rho}(MLR^{S})=1$ holds for any $S \in 2^{\omega}$ and any computable measurement system $B$. 

\end{defn}
By Remark \ref{rem:1} and by the above discussion on relativizations, we can also consider measurement of a state in non-computable measurement systems by using an appropriate oracle. We do not explore this here.

One may ask if we can build other computable examples of $\rho$s which are not q-MLR and are mR. We note that a straightforward modification of the proof of Theorem 3.4 yields a family of such $\rho$s. We do not provide all the details here for lack of space. 
 Let $h: \mathbb{N} \longrightarrow \mathbb{N}$ and $g: \mathbb{N} \longrightarrow (0,1)$ be computable, satisfying the following for some constants $\delta \in (0,1)$ and $F>0$: \[\prod_{n=5}^{\infty} (1- h(n)2^{-n})=0 , \prod_{n=5}^{\infty} (1- h(n)[2^{-n}-g(n)])= \delta,\]\[ \forall n, g(n) \leq 2^{-n} \text{ and } \prod_{n=5}^{\infty} \big[1- \dfrac{4g^{2}(n)h^{2}(n)}{(1-2g(n)h(n))^{2}}\big] =F .\] 
Let $\rho$ be defined as in the proof of the main Theorem but with $r_n$ replaced by $h(n)$ and with the $h(n)$ many entries on the extreme ends of the anti-diagonal of $d_n$ being equal to $g(n)$ instead of $2^{-n}$. Then, this $\rho$ is computable and mR (in fact, it is strong mR) and fails a q-MLT at order $\delta$.
\\
\\
We are working towards characterizing the set of states for which mR and q-MLR are equivalent. So far, we have shown that these notions are equivalent for states of the form $\otimes_{n=1}^{\infty} d$ for some computable density matrix $d$.\\
\\
One may imagine using a sequence of POVMs in \ref{defn:78} instead of a sequence orthonormal bases. We use orthonormal bases to conform with the  work  \cite{unpublished } on which ours is based, which uses projective measurements and not POVMs. It would also be interesting to replicate the approach of \cite{unpublished} using POVMs instead of projective measurements.
\section{Concluding comments}
Measuring a finite dimensional quantum system or a composite system of finitely many qubits is a pivotal concept in quantum information theory \cite{bookA}. It hence seems natural to consider defining a notion of measuring a state. Since measurement of a state yields an infinite sequence of bits, it is interesting to explore the relation between the randomness of the measured state and the randomness of the resulting sequence. This paper is motivated by these questions. The main result is that q-MLR is not equivalent to mR even for the computable states.
\subsection{Remark}
Intuitively, the non-equivalence of mR and q-MLR should not be surprising given that entanglement in composite systems cannot be detected by independent measurements of the subsystems. Let us elaborate on this remark. $\rho$ in \ref{thm:0000} is built up from $d_n$s where each $d_n$ has $r_n$ many entangled eigenvectors with non-zero eigenvalue and $r_n$ many entangled eigenvectors with zero eigenvalue. This inhomogeneity in the distribution of eigenvalues is solely due to these entangled eigenvectors (all the $2^{n}- 2r_n$ many non entangled eigenvectors of $d_n$ have the same  non-zero eigenvalues). A crucial part in showing that $\rho$ is non q-MLR was to use the inhomogeneous eigenvalue distribution to bound the size  $M_N$ (see \ref{eq:eig} in the proof of  \ref{lem:20}). Heuristically speaking, the the non-quantum randomness of $\rho$ was a reflection of the non-uniform eigenvalue distribution of $d_n$ which in turn was due to by the presence of entangled eigenvectors of $d_n$.   It is hence reasonable to expect that the quantum non-randomness of $\rho$, which stems from entanglement, cannot be captured by measurements in the sense of \ref{defn:78} using pure tensors (i.e. measuring each 2-dimensional subsystem independently).

\subsection{Notes} 
We have shown how to extract classical arithmetic randomness from a computable, non-quantum random sequence of qubits. It seems plausible that our results may prove to be relevant to the construction of quantum random number generators \cite{RevModPhys.89.015004}.
Abbott, Calude and Svozil have also studied classical bit sequences resulting from measuring a quantum system\cite{DBLP:phd/hal/Abbott15,DBLP:conf/birthday/AbbottCS15}. However, their notion of measurement is significantly different from ours. In contrast to our work which considers measurement of an infinite sequence of qubits, they studied the randomness of a sequence of bits generated by \emph{repeatedly measuring a finite dimensional} quantum system. They go on to apply this to quantum random number generators and their certification \cite{DBLP:phd/hal/Abbott15,DBLP:journals/mscs/AbbottCS14,DBLP:conf/birthday/AbbottCS15}.

\subsection{Acknowledgements}

I thank James Hanson for conjecturing that q-MLR is equivalent to mR when $d$ (see section 4) is four by four in the above. Joe Miller and Peter Cholak (independently) asked if there is a notion of `measuring a state' and if one gets a MLR bitstring from measuring a q-MLR state with respect to such a notion. These questions were one of the factors which led me to explore this area. Andr\'e Nies, whom I thank for introducing me to quantum algorithmic randomness, independently suggested that one might get a measure on Cantor space by `measuring' a state.  I thank Joe Miller, my thesis advisor, for his encouragement and guidance.

\nocite{*}
\bibliographystyle{eptcs}
\bibliography{references}
\end{document}